\documentclass{scrartcl}
\usepackage[utf8]{inputenc}

\usepackage{authblk}

\title{Mitigating Errors in Local Fermionic Encodings}
\author[1,3]{Johannes Bausch}
\author[1]{Toby Cubitt}
\author[1,2]{Charles Derby}
\author[1]{Joel Klassen\footnote{Authors listed alphabetically}}
\affil[1]{Phasecraft Ltd.}
\affil[2]{Department of Computer Science, University College London}
\affil[3]{Department of Applied Mathematics and Theoretical Physics, University of Cambridge}

\date{March 2020}

\usepackage{amsmath,amsthm,amsfonts}
\let\latextextsuperscript\textsuperscript
\AtBeginDocument{\let\textsuperscript\latextextsuperscript}
\usepackage{newtxtext,newtxmath,microtype,dsfont}

\input{braket.tex}

\usepackage[dvipsnames]{xcolor}
\definecolor{color1}{RGB}{230,57,70}
\definecolor{color2}{RGB}{29,53,87}
\definecolor{color3}{RGB}{69,123,157}
\usepackage[colorlinks=true,linkcolor=color1,urlcolor=color2,citecolor=color3,anchorcolor=green,breaklinks=true]{hyperref}

\usepackage[capitalise]{cleveref}
\crefname{lemma}{Lemma}{Lemmas}
\crefname{proposition}{Proposition}{Propositions}
\crefname{definition}{Definition}{Definitions}
\crefname{theorem}{Theorem}{Theorems}
\crefname{conjecture}{Conjecture}{Conjectures}
\crefname{corollary}{Corollary}{Corollaries}
\crefname{example}{Example}{Examples}
\crefname{section}{Section}{Sections}
\crefname{appendix}{Appendix}{Appendices}
\crefname{figure}{Fig.}{Figs.}
\crefname{equation}{Eq.}{Eqs.}
\crefname{table}{Table}{Tables}
\crefname{item}{Property}{Properties}
\crefname{remark}{Remark}{Remarks}

\newtheorem{theorem}{Theorem}

\newtheorem{proposition}[theorem]{Proposition}

\usepackage[backend=biber,style=alphabetic,sorting=none,doi=false,isbn=false,url=false,maxbibnames=20,maxcitenames=2]{biblatex}
\renewbibmacro{in:}{%
  \ifentrytype{article}
    {}
    {\printtext{\bibstring{in}\intitlepunct}}}

\newbibmacro{string+doi}[1]{\iffieldundef{doi}{#1}{\href{https://dx.doi.org/\thefield{doi}}{#1}}}
\DeclareFieldFormat{title}{\usebibmacro{string+doi}{\mkbibemph{#1}}}
\DeclareFieldFormat[article]{title}{\usebibmacro{string+doi}{\mkbibquote{#1}}}
\DeclareFieldFormat[incollection]{title}{\usebibmacro{string+doi}{\mkbibquote{#1}}} 

\renewcommand\paragraph[1]{\textit{\textbf{#1}}}

\DeclareMathOperator{\BigO}{O}

\usepackage{tikz}
\usetikzlibrary{decorations.markings}
\usepackage{subcaption}
\usepackage{mathtools}

\newcommand\HFH{H_\mathrm{FH}}
\newcommand\ii{\mathrm i}
\newcommand\ee{\mathrm e}
\newcommand\1{\mathds 1}
\newcommand\DK{LW}
\newcommand\DKfull{Low-Weight}
\DeclareDocumentCommand{\c}{ s }{%
    \IfBooleanTF{#1}{\overline}{}\gamma%
}

\DeclareDocumentCommand{\d}{ s }{%
    \IfBooleanTF{#1}{\overline}{}\mu%
}

\bibliography{refs}

\usepackage{todonotes}

\begin{document}

\maketitle
\thispagestyle{empty}
\enlargethispage{3cm}

\begin{abstract}
    Quantum simulations of fermionic many-body systems crucially rely on mappings from indistinguishable fermions to distinguishable qubits. The non-local structure of fermionic Fock space necessitates encodings that either map local fermionic operators to non-local qubit operators, or encode the fermionic representation in a long-range entangled code space.
    In this latter case, there is an unavoidable trade-off between two desirable properties of the encoding: low weight representations of local fermionic operators, and a high distance code space.
    
Here it is argued that despite this fundamental limitation, fermionic encodings with low-weight representations of local fermionic operators can still exhibit error mitigating properties which can serve a similar role to that played by high code distances. In particular when undetectable errors correspond to ``natural'' fermionic noise. We illustrate this point explicitly for two fermionic encodings: the Verstraete-Cirac encoding, and an encoding appearing in concurrent work by~\citeauthor{derbyKlassen2020}.

 In these encodings many, but not all, single-qubit errors can be detected. However we show that the remaining undetectable single-qubit errors map to local, low-weight fermionic phase noise. We argue that such noise is natural for fermionic lattice models.
This suggests that even when employing low-weight fermionic encodings, error rates can be suppressed in a similar fashion to high distance codes, provided one is willing to accept simulated natural fermionic noise in their simulated fermionic system. 
\end{abstract}
\section{Introduction and Summary of Results}
One of the most promising prospective applications of near term quantum computers is the simulation of physical systems.
Fermionic systems are of particular interest, due to their prevalence in quantum chemistry and materials science.
On conventional qubit-based architectures a necessary ingredient in these simulations is a method of encoding fermionic systems into the qubits of the quantum computer: a ``fermion to qubit encoding''.
It can often be advantageous to employ an encoding which represents commonly used fermionic operators as low-weight and geometrically local qubit operators. This is true for example in quantum time dynamics simulation, where higher weight Hamiltonian terms give rise to higher depth circuits. Fermionic operators that appear commonly in fermionic many-body Hamiltonians include hopping terms and number operators.

Many fermion to qubit encodings have been designed with this in mind~\cite{Verstraete2005a,Bravyi2002,setia2019superfast,steudtner2019square,jiang2018majorana}.
A fundamental obstacle in the design of these encodings is that, unlike a multi-qubit state space, which can be decomposed into the product of local state spaces, the fermionic state space inherently has a non-local structure; exchanging any two fermions introduces a relative phase of $\pi$.
Consquently, any encoding which aims to represent fermionic operators by local qubit operators must invariably resort to representing states in a delocalized fashion, e.g.\ in a highly entangled subspace of a multi-qubit system.

Stabilizer codes~\cite[Chapter~10.5]{nielsen_chuang_2010} provide a convenient way of constructing highly entagled subspaces which are nonetheless amenable to mathematical analysis. All known local fermionic encodings take the form of stabilizer codes \cite{Bravyi2002,Verstraete2005a,setia2019superfast,jiang2018majorana,steudtner2019square}. A stabilizer code takes an $N$ qubit system and restricts it to the simultaneous +1 eigenspace, called the code space, of a mutually commuting group of Pauli operators called the stabilizer group. The dimension of the code space is $2^{N-d}$ where $d$ is the size of the generating set of the stabilizer group. Operators that commute with elements of the stabilizer group will preserve the code space and are referred to as logical operators. Operators that do not commute take states out of the code space and thus can be detected by stabilizer measurements. Thus these operators are called detectable. If the stabilizer measurements yield a unique signature, then the operators are called correctable. For the purposes of reducing error rates it is preferable to engineer a code to detect or correct as many error operators as possible. For instance, there exists a generalised superfast encoding which has been shown to be capable of correcting all single-qubit errors, provided the system to simulate has an interaction graph of large-enough degree \cite{setia2019superfast}. It is therefore natural to consider to what extent the existing local fermionic encodings can correct or detect physical qubit errors.

However, good error correcting or detecting codes must have high weight logical operators, so that logical operators do not occur as errors. This is fundamentally at odds with the goal of having low-weight encoded fermionic operators. In this work we argue that, despite the apparent conflict between low-weight fermionic operators and error correction and detection, there can exist valuable error mitigating properties of fermionic encodings that do \emph{not} need to be sacrificed in the pursuit of low-weight fermionic operators. In particular, in the context of fermionic simulation of natural systems, one might tolerate---or even desire---some noise in the physical qubits, provided that this noise translates into ``natural'' fermionic noise in the simulated fermionic system. This will depend crucially on the choice of encoding, and on what fermionic operators the low-weight undetectable errors correspond to in the logical fermionic space.

This point has been made in a more general context in~\cite{Cubitt2017}. There, the authors prove that for any local Hamiltonian simulations, as defined rigorously in that paper, local physical noise in the simulator system corresponds to local noise in the system being simulated. They also show that such local Hamiltonian simulations can indeed be constructed; in fact, they show there exist simple, universal quantum Hamiltonians that are able to simulate \emph{any} target Hamiltonian, to arbitrary precision. The general theoretical results of~\cite{Cubitt2017} do not address specific natural noise models, nor do they address fermionic systems.

In this work, we demonstrate how this mapping of noise models appears in specific fermionic encodings of interest, for specific families of local noise models.
In particular, we consider two fermionic encodings: the Verstraete Cirac (VC) encoding ~\cite{Verstraete2005a}, and an encoding presented concurrently by two of the authors of this work ~\cite{derbyKlassen2020}, which we refer to as the \DKfull{} (\DK{}) encoding.
These encodings are particularly well-suited to simulating fermionic lattice models, and give rise to low-weight terms for Hamiltonians which are local on these lattices, such as the Fermi-Hubbard model.

For concreteness, we consider independent and identically-distributed (iid) local qubit noise on the qubits in the quantum computer, and study the effect of this on fermionic simulations. Though this is a simplistic---albeit commonly-used---theoretical model of noise in real quantum computers, the salient features of our results depend only on the locality of the noise model, rather than the specific form of the noise operators. Moreover, this simplistic model is a surprisingly good match to the noise observed in current hardware~\cite{Google_supremacy}.

We first note that for both the VC and \DK{} encodings, many weight-1 errors are detectable, i.e. do not commute with stabilizers.

We then show that almost all of the remaining \emph{un}detectable errors translate into natural local noise in the fermionic system, and we present strategies for mitigating the edge cases where they do not.
In particular, we show that for the VC encoding, almost all weight-1 Pauli errors correspond to local fermionic phase noise, except for one instance of a weight-1 error on a corner qubit, which corresponds to a single Majorana error.
For the \DK{} encoding, we show that almost all weight-1 Pauli errors also correspond to local fermionic phase noise, except for cases where again weight-1 errors at corner qubits can give rise to a single Majorana error.
Furthermore, we explore the potential for mitigating weight-2 errors in a similar fashion.

The first-order Majorana edge cases can give rise to unphysical fermionic states which violate parity superselection.
We consider two strategies for mitigating these errors.
One strategy for the \DK{} encoding is to introduce a slight modification to the lattice and the encoding at the corners, which results in single Majorana errors at the corners corresponding to weight-2 Pauli errors.
The other strategy is to employ global fermionic parity check stabilizer---which may be cheap to implement e.g.\ if stabilizers are only measured once at the end of the computation.

Low Pauli-weight operators and noise resilience are both features of fermionic encodings of particular relevance for near term quantum algorithms, where quantum circuits are run for a short period of time and error correction is unavailable; a faithful execution can thus only be guaranteed if the number of errors in that time period is $\ll 1$.
We expand on this point in the discussion.

\section{Preliminaries}\label{sec:preliminaries}

\subsection{Fermionic Lattice Models}\label{sec:lattice-models}
In nature, fermions are typically embodied as particles with continuous positions and momenta.
As such, the quantum states of fermionic systems admit their most general representations as functions over continuous degrees of freedom.
However, under many circumstances there may exist a discrete set of approximately or exactly orthogonal functions over which the natural fermionic wavefunctions admit a convenient decomposition.
This is true for example in the case of the atomic orbitals of atoms in a molecule or a crystal. One might also choose to approximately discretize a continuous coordinate system, in the hopes of retrieving an effective model that still captures essential features of the system under study.
It can thus be sufficient to represent fermionic systems by a discrete set of modes (or sites), and by interactions between these modes.
It is here where fermion-to-qubit encodings become relevant: they map a discrete set of fermionic modes to a discrete set of qubits.

The encodings we consider are particularly well-suited to lattice models, in which the fermionic modes are arranged in a lattice configuration, and the interactions are short range with respect to the lattice geometry.
These models are commonly found in materials science and condensed matter physics.
The archetypal toy lattice model for fermionic systems is the Fermi-Hubbard model, which typically lives on a square lattice. In its basic form, the Fermi-Hubbard model describes spin-1/2 fermions hopping on a square lattice, with an on-site interaction between two fermions (of opposite spin) occupying the same lattice siate. The Fermi-Hubbard Hamiltonian is given by
\begin{align}
    \HFH  =
    &-\sum_{\langle j,k\rangle}\sum_{\sigma\in\{\uparrow,\downarrow\}} t_{jk}\left(a^\dagger_{j,\sigma}a_{k,\sigma}+a^\dagger_{k,\sigma}a_{j,\sigma}\right)  \label{eq:FH-hopping}\\
    &+U\sum_j n_{j,\uparrow}n_{j,\downarrow}  \label{eq:FH-onsite}\\
    &+\sum_j \sum_{\sigma\in\{\uparrow,\downarrow\}} (\epsilon_j-\mu_j)n_{j,\sigma} \\
    &-\sum_{j} h_j(n_{j,\uparrow}-n_{j,\downarrow}),
    \label{eq:FH-external}
\end{align}
where $\langle j,k\rangle$ denotes adjacent sites on the lattice, with corresponding Dirac fermion creation and annihilation operators $\smash{a^\dagger_{i,\sigma},a_{i,\sigma}}$, and the fermion number (or density) operator $n_{j,\sigma}\coloneqq \smash{a_{j,\sigma}^\dagger a_{j,\sigma}}$ for a spin $\sigma\in\{\uparrow,\downarrow\}$.
The first term (\cref{eq:FH-hopping}) describes a nearest-neighbour hopping term with coupling strengths $t_{jk}$;
the second (\cref{eq:FH-onsite}) describes the on-site Coulomb repulsion between particles of different spin, the third and fourth terms are the energy from local potentials and magnetic fields, respectively.
Each site can either be occupied by a fermion, with corresponding odd parity---or be unoccupied, corresponding to an even parity.

In addition to fermionic creation and annihilation operators, it can often be convenient to consider Majorana operators, defined as
\begin{equation}\label{eq:majorana-identities}
    \c_j \coloneqq  a_{j}+a^\dagger_{j} \quad\text{and}\quad \c*_j \coloneqq  -\ii \left(a_{j}-a^\dagger_{j}\right).
\end{equation}
The Majorana operators are self-inverse and mutually anticommute. Application of one of the Majorana operators flips a state's fermion parity.

\subsection{The Verstraete-Cirac Encoding}
In order to simulate a fermionic Hamiltonian, such as the Fermi-Hubbard Hamiltonian, with qubits and qubit couplings, one requires a map between the fermion and qubit operators, also called an encoding.
Such an encoding will respect the commutation relations of the operators and allow the identification of fermionic states with qubit configurations. One encoding we examine is the Verstraete-Cirac encoding~\cite{Verstraete2005a} which we will briefly review here.

The Verstraete-Cirac (VC) encoding constitutes a stabilizer code, with the code space representing the fermionic Fock space. The VC encoding makes use of the Jordan-Wigner (JW) transform~\cite{Jordan1928,nielsen2005fermionic}, an encoding in which an ordering is imposed on the fermionic sites, and Majorana operators $\c_j$ and $\c*_j$ are encoded as qubit operators which include long strings of $Z$ Paulis:
\begin{equation}
 \tilde{\c}_j = \left(\prod_{i<j} Z_i \right)X_j \;,\; \tilde{\c*}_j = \left(\prod_{i<j} Z_i \right)Y_j .
\end{equation} 
  Throughout we will denote encoded fermionic operators with a tilde above the symbol. These strings appear most prominently for local operators acting on sites which are local according the lattice geometry, but non-consecutive according to the JW ordering.

More concretely, consider a fermionic lattice graph $G=(\mathsf{V}, \mathsf{E})$, with the vertices in $\mathsf{V}$, indexed by Latin indices $i,j,k$ corresponding to a JW ordering.
The JW ordering is conventionally chosen to be a Hamiltonian path over the primary sites, so that $i=j+1$ implies $(i,j) \in \mathsf{E}$. Consider a pair of sites $i<k \in \mathsf{V}$ which share an edge $(i,k) \in \mathsf{E}$ but which are not consecutive in the ordering (ie  $k\neq i+1$). In the JW encoding, local fermionic terms acting on such pairs of sites will exhibit long range strings. For example the hopping term is encoded as
\begin{equation}\label{eq:hopping}
\tilde{a}_{k}^\dagger \tilde{a}_{i} + \tilde{a}_{i}^\dagger \tilde{a}_{k}  =\frac{\ii}{2} ( \tilde{\c}_i \tilde{\c*}_k- \tilde{\c*}_i \tilde{\c}_k ) =  \frac{1}{2}\left(Y_{i}Y_{k}+X_{i}X_{k}\right)\left(\prod_{j=i+1}^{k-1}Z_j\right).
\end{equation}

 The principle behind the VC encoding is to avoid these long strings of $Z$ operators by introducing additional auxiliary fermionic sites (indexed in the ordering by $j'$) for every primary fermionic site $j$, and choosing a JW ordering which alternates between the primary and auxiliary fermionic sites, so that $j<j'<j+1$. All fermions are then encoded according to the JW transform. For clarity we will denote primary Majoranas by $\c_j$ and auxiliary Majoranas by $\d_{j'}$. Thus the encoded primary Majoranas are
 \begin{equation}
  \tilde{\c}_j = \left(\prod_{i<j} Z_i Z_{i'} \right)X_j \quad\text{and}\quad \tilde{\c*}_j = \left(\prod_{i<j} Z_i Z_{i'}\right)Y_j,
 \end{equation}
 and the encoded auxiliary Majoranas are
 \begin{equation}
  \tilde{\d}_j = \left(\prod_{i<j} Z_i Z_{i'} \right)Z_jX_{j'} \quad\text{and}\quad \tilde{\d*}_j = \left(\prod_{i<j} Z_i Z_{i'}\right)Z_jY_{j'}.
 \end{equation}

The stabilizers of the encoding are then defined as products of pairs of encoded Majoranas on the auxiliary sites (for example $\ii\tilde{\d}_{i'} \tilde{\d}_{k'}$ or $\ii\tilde{\d}_{i'} \tilde{\d*}_{k'}$ etc.\footnote{The factors of $\ii$ ensure they are valid stabilizers with $\pm 1$ eigenvalues, see \Cref{eq:VC_aux_mappings}.}), such that every auxiliary Majorana appears in exactly one such pairing. In this way the code space consists only of the primary fermionic degrees of freedom. Furthermore, these pairs are chosen such that for every edge $(i,k) \in \mathsf{E}$ for which $i$ and $k$ are not consecutive in the JW ordering, there is a stabilizer corresponding to a pair of Majoranas from the associated auxiliary sites, for example $\ii\tilde{\d}_{i'} \tilde{\d}_{k'}$.  $G$ must have max degree $4$ in order for this to be possible. See  \Cref{fig:Maj pairings} for how such a pairing can be chosen on a square lattice.

The motivation behind this choice of stabilizers is that the strings of $Z$ operators appearing in local fermionic operators on the primary sites can be cancelled by the strings of $Z$ operators appearing in the stabilizers, yielding a logically equivalent operator which does not have long range strings. For example in the VC encoding the hopping term in \Cref{eq:hopping} would take the modified form:
\begin{equation}
\tilde{a}_{k}^\dagger \tilde{a}_{i} + \tilde{a}_{i}^\dagger \tilde{a}_{k}= \frac{1}{2}\left(Y_{i}Y_{k}+X_{i}X_{k}\right)Z_{i'}\left(\prod_{j=i+1}^{k-1}Z_jZ_{j'}\right)
\end{equation}
However by multiplying this operator by, for example, the stabilizer operator
\begin{equation}
\ii\tilde{\d}_{i'} \tilde{\d}_{k'} = Y_{i'} \left(\prod_{j=i+1}^{k-1}Z_jZ_{j'}\right) Z_k X_{k'}
\end{equation}
one retrieves a logically equivalent operator that is local in the lattice geometry , i.e.
\begin{equation}
\ii\tilde{\d}_{i'} \tilde{\d}_{k'}\left(\tilde{a}_{k}^\dagger \tilde{a}_{i} + \tilde{a}_{i}^\dagger \tilde{a}_{k}\right)= \frac{1}{2}\left( Y_{i}X_{k}-X_{i}Y_{k}\right)X_{i'}X_{k'}.
\end{equation}

For ease of reference we include here how each auxiliary Majorana pairing is encoded.
If we assume $i<j$, then
\begin{equation}\label{eq:VC_aux_mappings}
\begin{aligned}[]
    \ii\tilde{\d}_i\tilde{\d}_j&= Y_{i'}\left(\prod_{k=i+1}^{j-1} Z_k Z_{k'}\right)Z_jX_{j'}, & &
    -\ii\tilde{\d*}_i\tilde{\d}_j&= X_{i'}\left(\prod_{k=i+1}^{j-1} Z_k Z_{k'}\right)Z_jX_{j'}, \\
    \ii\tilde{\d}_i\tilde{\d*}_j&= Y_{i'}\left(\prod_{k=i+1}^{j-1} Z_k Z_{k'}\right)Z_jY_{j'}, & &
    -\ii\tilde{\d*}_i\tilde{\d*}_j&= X_{i'}\left(\prod_{k=i+1}^{j-1} Z_k Z_{k'}\right)Z_jY_{j'}, \\
    & &-\ii\tilde{\d}_i\tilde{\d*}_i=Z_{i'}.
\end{aligned}
\end{equation}


\begin{figure}[t]
    \centering
    \begin{subfigure}[b]{0.4\textwidth}
        \begin{tikzpicture}[scale=1]

\begin{scope}[scale=1.5]
\draw[opacity=0.1,line width=11pt,cap=rect](0,0)--(3,0)--(3,1)--(0,1)--(0,2)--(3,2)--(3,3)--(0,3);
\end{scope}

\foreach \x in {0,1,2,3}{
    \foreach \y in {0,1,2,3}{
        \node at (1.5*\x,1.5*\y+0.25) [fill=none]{$\d$};
        \node at (1.5*\x,1.5*\y-0.25)[fill=none]{$\d*$};
        \node at (1.5*\x,1.5*\y) [shape=circle,fill=black,scale=0.4]{};
        }
    }

\foreach \x in {0,1.5,3,4.5}{
    \foreach \y in {0,1.5,3}{
        \draw[blue] (\x+0.25,\y+0.25)--(\x+0.25,\y+1.5-0.25) arc (0:180:0.25cm)--(\x-0.25,\y+0.25) arc (-180:0:0.25cm);
        }
    }

\foreach \x in {0,3}{
    \draw[blue] (\x,0)--(\x+1.5,0) arc (90:-90:0.25cm)--(\x,0-0.5) arc (270:90:0.25cm);
    \draw[blue] (\x,5)--(\x+1.5,5) arc (90:-90:0.25cm)--(\x,4.5) arc (270:90:0.25cm);
    }

\end{tikzpicture}
\caption{Even column number.}
    \end{subfigure}
    \qquad
    \begin{subfigure}[b]{0.4\textwidth}
    \centering
        \begin{tikzpicture}[scale=1]

\begin{scope}[scale=1.5]
\draw[opacity=0.1,line width=11pt,cap=rect](0,0)--(2,0)--(2,1)--(0,1)--(0,2)--(2,2)--(2,3)--(0,3);
\end{scope}

\foreach \x in {0,1,2}{
    \foreach \y in {0,1,2,3}{
        \node at (1.5*\x,1.5*\y+0.25) []{$\d$};
        \node at (1.5*\x,1.5*\y-0.25)[]{$\d*$};
        \node at (1.5*\x,1.5*\y) [shape=circle,fill=black,scale=0.4]{};
        }
    }

\foreach \x in {0,1.5}{
    \foreach \y in {0,1.5,3}{
        \draw[red] (\x+0.25,\y+0.25)--(\x+0.25,\y+1.5-0.25) arc (0:180:0.25cm)--(\x-0.25,\y+0.25) arc (-180:0:0.25cm);
        }
    }

\draw[red] (0,0)--(1.5,0) arc (90:-90:0.25cm)--(0,0-0.5) arc (270:90:0.25cm);
\draw[red] (0,5)--(1.5,5) arc (90:-90:0.25cm)--(0,4.5) arc (270:90:0.25cm);
\draw[red] (3+0.25,1.5+0.25)--(3+0.25,1.5+1.5-0.25) arc (0:180:0.25cm)--(3-0.25,1.5+0.25) arc (-180:0:0.25cm);
\draw[red] (3+0.25,4.5-0.25)--(3+0.25,4.5+0.25) arc (0:180:0.25cm)--(3-0.25,4.5-0.25) arc (-180:0:0.25cm);
\draw[red] (3+0.25,-0.25)--(3+0.25,0.25) arc (0:180:0.25cm)--(3-0.25,-0.25) arc (-180:0:0.25cm);
\draw[red] (3,3) arc (270:90:0.25cm)--(3.5,3.5) arc (90:0:0.25cm)--(3.75,1.5-0.25) arc (0:-90:0.25cm)--(3,1) arc (270:90:0.25cm)--(3.25,1.5) arc (-90:0:0.25cm)--(3.5,3-0.25) arc (0:90:0.25cm)--cycle;

\end{tikzpicture}
\caption{Odd column number.}
\label{fig:Maj pairings odd}
    \end{subfigure}
    \caption{Diagrams showing ways to pair up auxiliary Majoranas to form stabilizers in the Verstraete-Cirac encoding.
    In the case of an odd number of columns, only every other row pair is linked in the final column.
    However, this does not impact locality as the other pairs are adjacent with respect to the chosen ordering (i.e.\ the grey path), which means that their interactions can still be made geometrically local.}\label{fig:Maj pairings}
\end{figure}
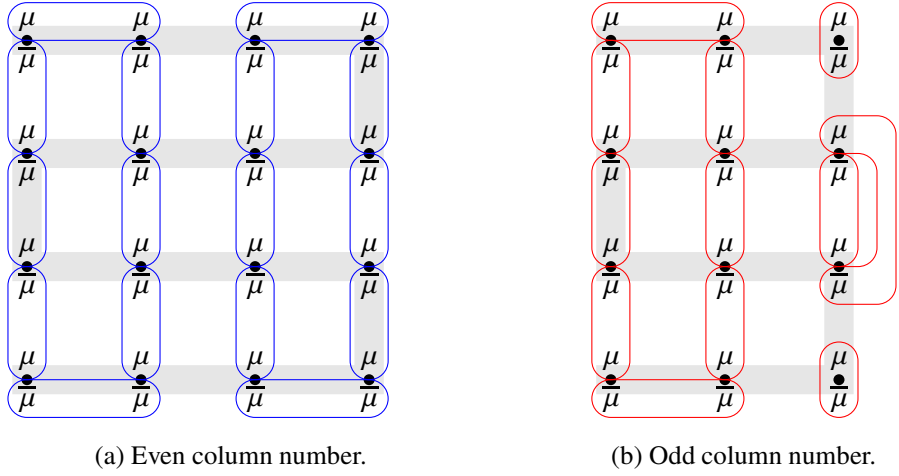

\subsection{The \DKfull{} Encoding}
The \DK{} encoding presented in~\cite{derbyKlassen2020} is also a stabilizer code, with the code space representing the fermionic Fock space.
The principle behind this encoding is related to that employed by Bravyi and Kitaev in their superfast encoding~\cite{Bravyi2002}, in which the stabilizers are associated with fermionic loop operators that must necessarily act as the identity in the fermionic Fock space, this principle is also behind the encodings in \cite{jiang2018majorana,setia2019superfast}. 
 In the case of a square lattice, the \DK{} encoding employs a qubit on every vertex and a qubit on every \emph{odd} face, according to a checker-board even-odd labelling of the faces.
The edges are also given an orientation so that they circulate clockwise or anticlockwise around \emph{even} faces.

The fermionic edge operators read $E_{ij}\coloneqq - \ii \c_i \c_j$ for all adjacent sites on the lattice, and they are dependent on the ordering of their indices, e.g.\ $E_{ji}=- \ii \c_j \c_i = \ii \c_i \c_j=-E_{ij}$.
In the square lattice encoding every edge $(i,j)$ has a unique odd face adjacent to it, whose qubit we label by $f(i,j)$.

The encoded edge operators $\tilde{E}_{ij}$ for an edge $(i,j)$ where $i$ points to $j$ are then defined to be
\begin{equation}
\tilde{E}_{ij} \coloneqq   \left\{ \begin{array}{rl}  X_i Y_j X_{f(i,j)} & \textrm{ if $(i,j)$ is oriented downwards} \\
-X_i Y_j X_{f(i,j)} & \textrm{ if $(i,j)$ is oriented upwards}  \\
X_i Y_j Y_{f(i,j)} & \textrm{ if $(i,j)$ is horizontal} \end{array} \right.
\end{equation}
Edge operators for a pair of indices $(i,j)$ \emph{against} the direction of an arrow are defined as $\tilde{E}_{ji} \coloneqq  - \tilde{E}_{ij}$.
For those edges on the boundary which are not adjacent to an odd face, the Pauli operator meant to be acting on the non-existent face qubit is omitted.
The graph orientation, face qubit placement, and forms of the edge operators are illustrated in \cref{fig:Weight3Scheme}.

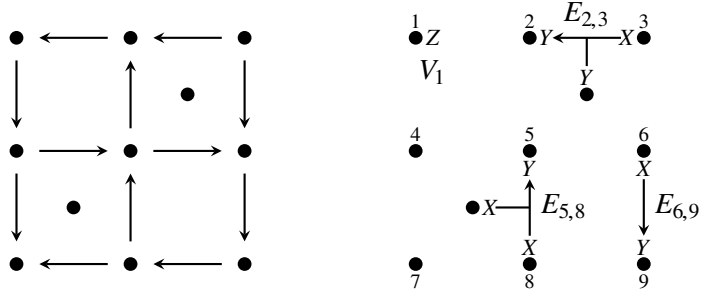
\begin{figure}[t]
\begin{center}
\begin{tikzpicture}[scale=1.5,>=stealth,thick]
\foreach \x in {0,1,2}{
    \foreach \y in {0,1,2}{
        \node at (\x,\y)[circle,fill=black,scale=0.5]{};
        }
    }

\node at (0.5,0.5)[circle,fill=black,scale=0.5]{};
\node at (1.5,1.5)[circle,fill=black,scale=0.5]{};

\foreach \x in {0,1}{
    \draw[->] (0.2+\x,1)--(0.8+\x,1);
    \draw[->] (1,0.2+\x)--(1,0.8+\x);
    \foreach \y in {0,2}{
        \draw[<-] (0.2+\x,\y)--(0.8+\x,\y);
        \draw[<-] (\y,0.2+\x)--(\y,0.8+\x);
        }
    }

\begin{scope}[shift={(3.5,0)}]
\foreach \x in {0,1,2}{
    \foreach \y in {0,1,2}{
        \node at (\x,\y)[circle,fill=black,scale=0.5]{};
        }
    }

\node at (.15,2)[scale=0.8]{$Z$};
\node at (.15,2-0.3)[]{$V_1$};

\foreach \x in {1,2,3}{
    \node at (\x-1,2.15)[scale=0.7]{$\x$};
    }
\foreach \x in {4,5,6}{
    \node at (\x-4,1.15)[scale=0.7]{$\x$};
    }
\foreach \x in {7,8,9}{
    \node at (\x-7,-.15)[scale=0.7]{$\x$};
    }

\node at (0.5,0.5)[circle,fill=black,scale=0.5]{};
\node at (1.5,1.5)[circle,fill=black,scale=0.5]{};

\node at (1.5,2.2)[]{$E_{2,3}$};
\node at (1.15,2)[scale=0.8]{$Y$};
\node at (2-0.15,2)[scale=0.8]{$X$};
\node at (1.5,1.65)[scale=0.8]{$Y$};
\draw[->] (2-0.2,2)--(1.2,2)[];
\draw[] (1.5,2)--(1.5,1.75)[];

\node at (1.3,0.5)[]{$E_{5,8}$};
\node at (1,0.15)[scale=0.8]{$X$};
\node at (1,1-0.15)[scale=0.8]{$Y$};
\node at (0.65,0.5)[scale=0.8]{$X$};
\draw[->] (1,0.25)--(1,1-0.25)[];
\draw[] (1,0.5)--(0.7,0.5)[];

\node at (2.3,0.5)[]{$E_{6,9}$};
\node at (2,0.15)[scale=0.8]{$Y$};
\node at (2,1-0.15)[scale=0.8]{$X$};
\draw[<-] (2,0.25)--(2,1-0.25)[];
\end{scope}
\end{tikzpicture}
\end{center}
\caption{(Left) The unit cell of the encoding. Dots correspond to qubits, and arrows correspond to edge orientations. (Right) Examples of encoded edge and vertex operators based on this layout.}
\label{fig:Weight3Scheme}
\end{figure}
A cycle $p = \{p_1, p_2,\ldots\}$ is a sequence of fermionic sites such that $(p_i, p_{i+1})$ is an edge of the lattice, and $p_1 = p_{\vert p\vert}$. Given a cycle $p$, the fermionic loop operator takes the form
\begin{equation}\label{eq:req4}
  \ii^{\vert p\vert-1} \prod_{i=1}^{\vert p\vert-1} E_{p_i p_{i+1}} =1.
\end{equation}
Within the encoding, these loop operators correspond to stabilizers.
The cycle basis of a graph are those cycles with which one may construct all other cycles via composition, where here the composition of two cycles is a new cycle over the disjunctive union \footnote{$a \cup b - a \cap b$} of their edges. It suffices for the generators of the stabilizers to correspond to a cycle basis of the lattice. For simplicity the cycle basis may be chosen to be the cycles around square faces. For a cycle $p$ around an ``odd'' face, we have
\begin{equation}
  \ii^{\vert p\vert-1} \prod_{i=1}^{\vert p\vert-1} \tilde{E}_{p_i p_{i+1}} =1,
\end{equation}
so these need not be included as generators.
For cycles around the even faces, let $1,2,3,4$ be the sequence of vertices around an even loop, starting with $1$ from the upper left hand corner and proceeding clockwise.
Then the corresponding stabilizer generator takes the form
\[
  \ii^{4} \tilde{E}_{12} \tilde{E}_{23} \tilde{E}_{34} \tilde{E}_{41}   =  Z_1 Z_2 Z_3 Z_4 Y_{f(1,2)} X_{f(2,3)} Y_{f(3,4)} X_{f(4,1)}
\]
as is illustrated in \cref{fig:Stabilizer}.

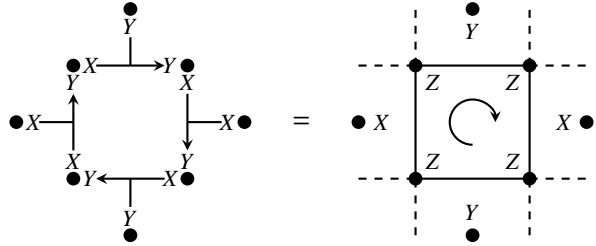
\begin{figure}[t]
    \begin{center}
    \begin{tikzpicture}[scale=1.5,>=stealth,thick]

\foreach \x in {0,1}{
    \foreach \y in {0,1}{
        \node at (\x,\y)[circle,fill=black,scale=0.5]{};
        \node at (0.5+\x-\y,-0.5+\x+\y)[circle,fill=black,scale=0.5]{};

        }
    }

\node at (0.15,0)[scale=0.8]{$Y$};
\node at (1-0.15,0)[scale=0.8]{$X$};
\node at (0.5,1.65-2)[scale=0.8]{$Y$};
\draw[->] (1-0.2,0)--(0.2,0)[];
\draw[] (.5,0)--(0.5,1.75-2)[];

\begin{scope}[rotate=180,shift={(-1,-1)}]
\node at (0.15,0)[scale=0.8]{$Y$};
\node at (1-0.15,0)[scale=0.8]{$X$};
\node at (0.5,1.65-2)[scale=0.8]{$Y$};
\draw[->] (1-0.2,0)--(0.2,0)[];
\draw[] (.5,0)--(0.5,1.75-2)[];
\end{scope}

\node at (0,0.15)[scale=0.8]{$X$};
\node at (0,1-0.15)[scale=0.8]{$Y$};
\node at (0.65-1,0.5)[scale=0.8]{$X$};
\draw[->] (0,0.25)--(0,1-0.25)[];
\draw[] (0,0.5)--(0.7-1,0.5)[];

\begin{scope}[rotate=180,shift={(-1,-1)}]
\node at (0,0.15)[scale=0.8]{$X$};
\node at (0,1-0.15)[scale=0.8]{$Y$};
\node at (0.65-1,0.5)[scale=0.8]{$X$};
\draw[->] (0,0.25)--(0,1-0.25)[];
\draw[] (0,0.5)--(0.7-1,0.5)[];
\end{scope}

\node at (2,0.5)[]{$=$};

\begin{scope}[shift={(3,0)}]
\draw (0,0)--(1,0)--(1,1)--(0,1)--cycle;

\foreach \x in {0,1}{
    \foreach \y in {0,1}{
        \node at (\x,\y)[circle,fill=black,scale=0.5]{};
        \node at (0.5+\x-\y,-0.5+\x+\y)[circle,fill=black,scale=0.5]{};
        \node at (0.15+0.7*\x,0.15+0.7*\y)[scale=0.8]{$Z$};

        \draw[dashed] (-0.5+1.5*\x,\y)--(1.5*\x,\y);
        \draw[dashed] (\y,-0.5+1.5*\x)--(\y,1.5*\x);
        }
    }

\node at (0.5,1.7-2)[scale=0.8]{$Y$};
\node at (0.5,1.5-0.2)[scale=0.8]{$Y$};
\node at (-0.3,0.5)[scale=0.8]{$X$};
\node at (1.3,0.5)[scale=0.8]{$X$};

\draw[->] (0.5,0.3) arc (270:0:0.2);

\end{scope}
\end{tikzpicture}
    \end{center}
    \caption{Non-trivial loop stabilizer of the encoding. If a loop is on the boundary of the lattice, omit any Pauli operators for which no face qubits exist.}
    \label{fig:Stabilizer}
\end{figure}

\subsection{Natural Noise on Fermionic Lattice Models}\label{sec:natural-noise}
A common fermionic lattice model is one in which lattice sites correspond to atomic positions. These atomic positions are often considered to be fixed.
However, one may consider the possibility of phonons in the lattice of atoms, and how these phonons couple to the electrons as a source of noise.
To first order, this coupling is dominated by low energy acoustical modes ~\cite{fedorov2004decoherence,bruus2002introduction}, with interaction Hamiltonian
\begin{equation}
    H_\text{int}=\frac{1}{V}\sum_{\mathbf{k},\sigma}\sum_\mathbf{q}g_\mathbf{q}a^\dagger_{\mathbf{k}+\mathbf{q},\sigma}a_{\mathbf{k},\sigma}\left(b_\mathbf{q}+b^\dagger_{-\mathbf{q}}\right),
\end{equation}
where $a^{(\dagger)}_{\mathbf{k}}$ and $b^{(\dagger)}_{\mathbf{k}}$ are, respectively, the annihilation and creation operators for fermions and phonons with momentum $\mathbf{k}$.

For a thermal bosonic bath, the effective noise model of this interaction on the fermionic system is spontaneous hopping of fermions into different momentum modes.
In the position basis this translates into dephasing noise~\cite{melnikov2016quantum}, since motion in momentum space corresponds to phase shifts in position space. The fermionic dephasing operator is:
\begin{equation}
    \left(1-2N_j\right)=-\ii\c_j\c*_j.
\end{equation}
We present a more thorough account of the derivation of this noise model in \cref{appendix:phaseNoise}.

In this paper we show that the first order noise experienced by VC and \DK{} encoded fermionic systems is dominated by noise of this type, sometimes exclusively.

\section{Mapping Physical Errors to Logical Errors}
In this section we present the main technical component of this work. We show that in both the VC encoding (\cref{sec:error-VC}) and the \DK{} encoding (\cref{sec:error-W3}) all weight-1 qubit errors fall into one of three categories: detectable errors; errors that correspond to mode-weight-1 phase noise; and errors that correspond to mode-weight-1 Majorana operators. The \emph{mode-weight} of a fermionic operator counts the number of fermionic modes acted on non-trivially (not to be confused Pauli weight). Thus, aside from the Majorana errors, all undetectable weight-1 errors correspond to low-mode-weight, local, and arguably natural fermionic noise.

Although the Majorana errors can only occur on very few sites, they can take a fermionic state into one which violates parity superselection.
In \cref{sec:ParityError} we present some novel techniques for avoiding or detecting these errors.

Our analysis here is motivated by the assumption that weight-1 Pauli noise dominates in the given noise model.
This is consistent with the most commonly studied qubit noise model, iid depolarising noise, amongst others.

\subsection{Verstraete-Cirac Encoding}\label{sec:error-VC}
To analyse how iid local qubit noise affects encoded fermionic states, we must understand how individual $X$, $Y$ and $Z$ Pauli errors on the qubits translate to the encoded fermionic Fock space of the VC encoding.
To this end, we first make the following observation:

\begin{proposition}[Undetectable VC Errors]\label{lemma:VC-undetectable}
In the VC encoding, all undetectable qubit errors correspond to either fermionic operators acting exclusively on primary fermionic sites, stabilizer operators, or products of operators of these two types.
\end{proposition}
\begin{proof}
This is a straightforward recasting of the fact that undetectable errors in quantum stabilizer codes correspond to logical operators, stabilizers, or products thereof \cite{gottesman1997stabilizer}. In the VC encoding the logical operators correspond to operators acting exclusively on the primary fermionic sites.
\end{proof}

For now, we focus on Pauli weight-1 errors only---i.e.\ the case where there is a \emph{single} such error occurring at one point within the qubit lattice.
There are six possible types of weight-1 Pauli errors: $X_i$, $Y_i$ or $Z_i$ on a primary qubit $i$, or $X_{i'}$, $Y_{i'}$, or $Z_{i'}$ on an ancillary qubit $i'$.

\begin{theorem}[Undetectable Weight-1 Errors in VC Encoding]\label{lem:VC weight-1}
In the VC encoding, the only non-trivial, undetectable, Pauli weight-1 errors are:
\begin{enumerate}
    \item $Z$ operators on primary qubits, which map to mode-weight-1 errors $Z_j \mapsto -\ii \c_j \c*_j$.
    \item $X$ or $Y$ errors on the first primary qubit according to the chosen fermion ordering, which map to $X_1 \mapsto \c_1$, $Y_1 \mapsto \c*_1$, respectively.
\end{enumerate}
\end{theorem}
\begin{proof}
Consider the forms of single qubit errors in the fermion picture (up to irrelevant global phases):
\begin{align*}
X_j&\longmapsto \left(\prod_{i<j}\c_i\c*_i\d_i\d*_i\right)\c_j, \quad & X_{j'}&\longmapsto \left(\prod_{i<j}\c_i\c*_i\d_i\d*_i\right)\c_j\c*_j\d_j, \\
Y_j&\longmapsto\left(\prod_{i<j}\c_i\c*_i\d_i\d*_i\right)\c*_j, \quad & Y_{j'}&\longmapsto \left(\prod_{i<j}\c_i\c*_i\d_i\d*_i\right)\c_j\c*_j\d*_j, \\
Z_j&\longmapsto\c_j\c*_j, \quad & Z_{j'}&\longmapsto\d_j\d*_j.
\end{align*}

\paragraph{Auxiliary Qubits.}
$X_{j'}$ and $Y_{j'}$ operators on the auxiliary qubits map to \emph{odd} products of auxiliary Majoranas, whereas all stabilizer operators are \emph{even} products of these. Thus by \cref{lemma:VC-undetectable} they are detectable.
An operator $Z_{j'}$ acting on an auxiliary qubit maps to an opertor acting only on auxiliary fermions. Thus by \cref{lemma:VC-undetectable} it can only be undetectable if it is a stabilizer. It may or may not be a stabilizer depending on the pairing chosen---see e.g.\ \cref{fig:Maj pairings odd} which shows a pairing where $Z_{j'}$ is a stabilizer on two corner sites.
Thus it is either a detectable error (i.e.\ not a stabilizer), or a trivial undetectable error (i.e.\ a stabilizer).

\paragraph{Primary Qubits.}
An operator $Z_j$ on a primary qubit maps to a mode weight-1 operator on one fermionic site only. Since there are no auxiliary Majorana operators involved, this error is undetectable by \cref{lemma:VC-undetectable}.

This leaves $X_j$ and $Y_j$ operators on primary qubits.
Each of these maps to a fermionic operator which includes the product of all auxiliary Majoranas on auxiliary sites $\le j'$ in the ordering, and includes no auxiliary Majoranas on those auxiliary sites greater than or equal to $j'$ in the ordering. 

Recall that the VC encoding pairs up all Majoranas on adjacent, non-consecutive auxiliary sites. Consider the product of all auxiliary Majorana operators on auxiliary sites $\le j'$. If any of these sites is adjacent to a non-consecutive auxiliary site $> j'$, it isn't paired up correctly, and cannot be a stabilizer. Thus $X_j$ and $Y_j$ can only be stabilizer if none of the auxiliary sites $\le j'$ is adjacent to an auxiliary site $> j'$.
However on a 2D grid, regardless of ordering, there is only one case where this can occur: when the $X_j$ or $Y_j$ acts on the first site in the ordering, where there are no prior sites. In this case the errors $X$ and $Y$ map to single Majorana operators on the first site.
\end{proof}

To summarise, all undetectable weight-1 errors are encoded as mode-weight-1 fermionic operators acting on the corresponding fermionic sites, or are encoded as the identity. Almost all of these mode-weight-1 operators are local fermionic phase noise, except for two which act on a single mode as Majorana operators.

\subsection{\DK{} Encoding} \label{sec:error-W3}

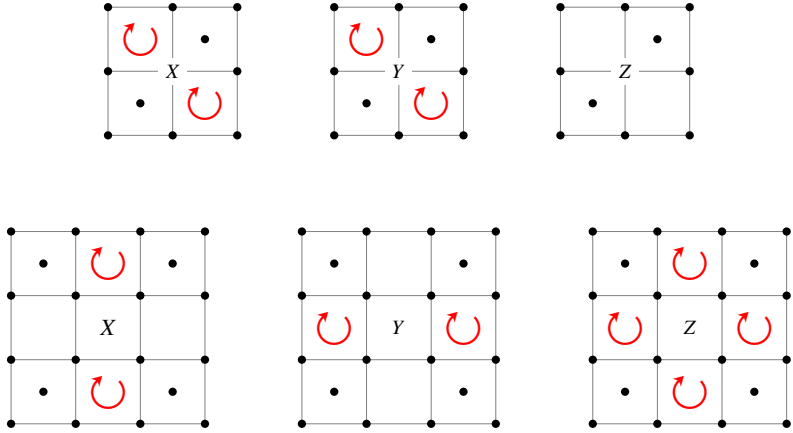
\begin{figure}[t]
    \centering
    \begin{tikzpicture}[scale=0.85]
\draw[step=1cm,gray,very thin] (0,0) grid (3,3);

\foreach \x in {0,1,2,3}{
    \foreach \y in {0,1,2,3}{
        \node at (\x,\y)[fill=black,circle,scale=0.3]{};
        }
    }
\foreach \x in {0,2}{
    \foreach \y in {0,2}{
        \node at (\x+0.5,\y+0.5)[fill=black,circle,scale=0.3]{};
        }
    }
\node at (1.5,1.5)[fill=white,scale=0.8]{$X$};
\node at (1.5,0.5)[text=red,scale=1.6]{$\circlearrowright$};
\node at (1.5,2.5)[text=red,scale=1.6]{$\circlearrowright$};

\begin{scope}[shift={(4.5,0)}]
\draw[step=1cm,gray,very thin] (0,0) grid (3,3);
\foreach \x in {0,1,2,3}{
    \foreach \y in {0,1,2,3}{
        \node at (\x,\y)[fill=black,circle,scale=0.3]{};
        }
    }
\foreach \x in {0,2}{
    \foreach \y in {0,2}{
        \node at (\x+0.5,\y+0.5)[fill=black,circle,scale=0.3]{};
        }
    }
\node at (1.5,1.5)[fill=white,scale=0.7]{$Y$};
\node at (0.5,1.5)[text=red,scale=1.6]{$\circlearrowright$};
\node at (2.5,1.5)[text=red,scale=1.6]{$\circlearrowright$};
\end{scope}

\begin{scope}[shift={(9,0)}]
\draw[step=1cm,gray,very thin] (0,0) grid (3,3);
\foreach \x in {0,1,2,3}{
    \foreach \y in {0,1,2,3}{
        \node at (\x,\y)[fill=black,circle,scale=0.3]{};
        }
    }
\foreach \x in {0,2}{
    \foreach \y in {0,2}{
        \node at (\x+0.5,\y+0.5)[fill=black,circle,scale=0.3]{};
        }
    }

\node at (1.5,1.5)[fill=white,scale=0.7]{$Z$};

\node at (0.5,1.5)[text=red,scale=1.6]{$\circlearrowright$};
\node at (2.5,1.5)[text=red,scale=1.6]{$\circlearrowright$};
\node at (1.5,0.5)[text=red,scale=1.6]{$\circlearrowright$};
\node at (1.5,2.5)[text=red,scale=1.6]{$\circlearrowright$};
\end{scope}

\begin{scope}[shift={(1.5,4.5)}]
\draw[step=1cm,gray,very thin] (0,0) grid (2,2);
\foreach \x in {0,1,2}{
    \foreach \y in {0,1,2}{
        \node at (\x,\y)[fill=black,circle,scale=0.3]{};
        }
    }
\node at (0.5,0.5)[fill=black,circle,scale=0.3]{};
\node at (1.5,1.5)[fill=black,circle,scale=0.3]{};
\node at (1,1)[fill=white,scale=0.7]{$X$};
\node at (1.5,0.5)[text=red,scale=1.6]{$\circlearrowright$};
\node at (0.5,1.5)[text=red,scale=1.6]{$\circlearrowright$};

\begin{scope}[shift={(3.5,0)}]
\draw[step=1cm,gray,very thin] (0,0) grid (2,2);
\foreach \x in {0,1,2}{
    \foreach \y in {0,1,2}{
        \node at (\x,\y)[fill=black,circle,scale=0.3]{};
        }
    }
\node at (0.5,0.5)[fill=black,circle,scale=0.3]{};
\node at (1.5,1.5)[fill=black,circle,scale=0.3]{};
\node at (1,1)[fill=white,scale=0.7]{$Y$};
\node at (1.5,0.5)[text=red,scale=1.6]{$\circlearrowright$};
\node at (0.5,1.5)[text=red,scale=1.6]{$\circlearrowright$};
\end{scope}

\begin{scope}[shift={(7,0)}]
\draw[step=1cm,gray,very thin] (0,0) grid (2,2);
\foreach \x in {0,1,2}{
    \foreach \y in {0,1,2}{
        \node at (\x,\y)[fill=black,circle,scale=0.3]{};
        }
    }
\node at (0.5,0.5)[fill=black,circle,scale=0.3]{};
\node at (1.5,1.5)[fill=black,circle,scale=0.3]{};
\node at (1,1)[fill=white,scale=0.7]{$Z$};
\end{scope}

\end{scope}
\end{tikzpicture}
    \caption{Graphical representation of the syndromes of single qubit errors on the \DK{} encoding on vertex qubits (upper) and face qubits (lower). The red circlular arrows represent the loop stabilizers that anticommute with the given error. $Z$ on a vertex qubit has no syndrome as it is a logical operator. For boundary cases where a highlighted stabilizer doesn't exist, it is omitted from the syndrome. Note that errors on vertex qubits can occur in contexts where the odd and even faces are inverted from how they are illustrated here. In this case the syndrome is similarly inverted.}
    \label{fig:weight-3 syndromes}
\end{figure}

In this section we argue that for the \DK{} encoding all Pauli weight-1 undetectable errors correspond to mode-weight 1 fermionic noise.

\begin{theorem}\label{lem:\DK{} weight-1}
In the \DK{} encoding for square lattices, the only non-trivial, undetectable, Pauli weight-1 errors are:
\begin{enumerate}
    \item $Z$ operators on primary qubits, which map to mode-weight-1 fermionic phase errors $Z_j \mapsto -\ii \c_j \c*_j$.
    \item In the case where the full fermionic space is encoded, $X$ or $Y$ errors on those corners adjacent only to an odd face, which map to single Majorana or Majorana hole operators, depending on the choice of convention.\footnote{For example $X_i \mapsto \c_i, \; Y_i \mapsto \c*_i$ or $X_i = \c_i \prod_j (- i \c_j \c*_j), \; Y_i \mapsto \c*_i \prod_{j} (-i \c_j \c*_j)$. For a fixed fermionic parity, the operator $\prod_{j} (-i \c_j \c*_j)$ is a good quantum number equal to $\pm1$, and so fermionic hole operators can be thought of as mode-weight-1 fermionic operators.}
\end{enumerate}
\end{theorem}
\begin{proof}
There can be $X$, $Y$ or $Z$ errors, either on the lattice face qubits (auxiliary qubits), or the vertices (primary qubits). We address them separately.

\paragraph{Auxiliary Qubits.} By inspection of the syndromes of single qubit errors shown in \cref{fig:weight-3 syndromes}, it is evident that Pauli weight 1 errors on face qubits always have a syndrome and are thus detectable.

\paragraph{Primary Qubits.} Any $Z$ errors on vertex qubis are undetectable as they have no syndrome; they correspond to fermionic phase errors.
On the other hand, $X$ and $Y$ errors on vertex qubits \emph{are} detectable, as they induce syndromes on two diagonally offset loop stabilizers. However there may be corners of the lattice which are not touching a stabilizer (i.e.\ a corner which only touches an odd face). $X$ or $Y$ errors on those vertex qubits correpond to single Majorana or Majorana hole operators. These are parity switching errors, so are not possible if only the even fermionic subspace is encoded.
\end{proof}

\subsection{Mitigating Parity Switching Errors} \label{sec:ParityError}

Parity switching errors, such as single Majoranas, can lead to violations of parity superselection in the encoded fermionic system.
One strategy to mitigate these erros, up to first order, is to measure the parity of the fermionic system as a stabilizer.
This parity stabilizer is given by the product of vertex operators at every fermionic site, which in our case corresponds to the product of $Z$ operators on every primary vertex qubit. However if one is performing non-destructive and coherent stabilizer measurements, or is interested in measuring observables that do not commute with the vertex operators, then such a stabilizer can be very costly to measure coherently.

In the following sections, we propose alternative methods for dealing with weight-1 parity switching errors in the VC and \DK{} encodings by introducing only minor modifications to the constructions of the respective encodings.

\subsubsection{Verstraete-Cirac Encoding}
The only parity switching errors at first order are single Majoranas on the first site. These can be made second order by a simple change in site ordering. If the first primary and first auxiliary site are switched in the ordering---i.e.\ such that the ordering becomes $1',1,2,2',3,3',4,4',\ldots$---and the code is constructed as usual, then the qubit representation of these errors will become weight-2 without affecting the weight of local interaction terms.

\subsubsection{\DK{} Encoding}
An even $\times$ even lattice can avoid parity switching errors altogether by choosing the appropriate checkerboard pattern of the auxiliary qubits such that none are placed in corner faces, permitting only the representation of parity-preserving fermionic operators. Even by odd and odd by odd lattices always have auxiliary qubits in two corners and so will always permit parity switching errors.

If the two corner fermionic sites adjacent to the odd faces are removed as shown in \cref{fig:CornerShaving}, then the parity switching errors correspond to Pauli weight-2 errors---i.e.\ the single Majorana or hole operator---has a weight-2 qubit representation.
One can preserve most of the structure of the corner, by introducing a new diagonal edge operator connecting those vertex qubits which had previously been connected to the removed site.
To ensure the correct anti-commutation relations, this new edge operator acts with a $Z$ on the face qubit and will act with the same Pauli operator on its incident vertex qubits as the two edge operators bounding the odd face (see \cref{fig:CornerShaving}).
The cycle operator formed by these three edge operators is the identity, provided the correct sign convention is chosen for the diagonal edge.

Single Majorana (or Majorana hole) operators may be added on either of the sites which were previously adjacent to the removed corner site, by applying a weight-2 Pauli operator on the corresponding vertex qubit and on the face qubit.
Note that these operators anti-commute with all incident edge operators and the vertex operator on that site.

\begin{figure}[t]
\begin{center}
\begin{tikzpicture}[scale=1.5,>=stealth,thick]

\foreach \x in {0,1}{
    \foreach \y in {0,1}{
        \node at (\x,\y)[fill=black,circle,scale=0.5]{};
        }
    }
\node at (0.5,0.5)[fill=black,circle,scale=0.5]{};

\node at (0.15,0)[scale=0.8]{$Y$};
\node at (0.85,0)[scale=0.8]{$X$};
\node at (0.5,0.35)[scale=0.8]{$Y$};
\draw[->] (1-0.2,0)--(0.2,0)[];
\draw[] (0.5,0)--(0.5,0.25)[];

\begin{scope}[rotate=180,shift={(-1,-1)}]
\node at (0.15,0)[scale=0.8]{$Y$};
\node at (0.85,0)[scale=0.8]{$X$};
\node at (0.5,0.35)[scale=0.8]{$Y$};
\draw[->] (1-0.2,0)--(0.2,0)[];
\draw[] (0.5,0)--(0.5,0.25)[];
\end{scope}

\node at (1,0.15)[scale=0.8]{$X$};
\node at (1,0.85)[scale=0.8]{$Y$};
\node at (.65,0.5)[scale=0.8]{$X$};
\draw[->] (1,0.25)--(1,0.75)[];
\draw[] (1,0.5)--(.7,0.5)[];

\begin{scope}[rotate=180,shift={(-1,-1)}]
\node at (1,0.15)[scale=0.8]{$X$};
\node at (1,0.85)[scale=0.8]{$Y$};
\node at (.65,0.5)[scale=0.8]{$X$};
\draw[->] (1,0.25)--(1,0.75)[];
\draw[] (1,0.5)--(.7,0.5)[];
\end{scope}

\draw[dashed] (1.2,1)--(1.5,1);
\draw[dashed,<-] (1.2,0)--(1.5,0);
\draw[dashed] (0,-0.2)--(0,-0.5);
\draw[dashed,<-] (1,-0.2)--(1,-0.5);

\node at (1.6,0.5){$\rightarrow$};

\begin{scope}[shift={(2,0)}]
\foreach \x in {0,1}{
    \foreach \y in {0,1}{
        \node at (\x,\y)[fill=black,circle,scale=0.5]{};
        }
    }
\node at (0.5,0.5)[fill=black,circle,scale=0.5]{};
\node at (0,1)[fill=white,circle,scale=0.4]{};

\node at (0.15,0)[scale=0.8]{$Y$};
\node at (0.85,0)[scale=0.8]{$X$};
\node at (0.5,0.35)[scale=0.8]{$Y$};
\draw[->] (1-0.2,0)--(0.2,0)[];
\draw[] (0.5,0)--(0.5,0.25)[];

\node at (1,0.15)[scale=0.8]{$X$};
\node at (1,0.85)[scale=0.8]{$Y$};
\node at (.65,0.5)[scale=0.8]{$X$};
\draw[->] (1,0.25)--(1,0.75)[];
\draw[] (1,0.5)--(.7,0.5)[];

\node at (0,0.15)[scale=0.8]{$Y$};
\node at (0.85,1)[scale=0.8]{$Y$};
\node at (0.4,0.6)[scale=0.8]{$Z$};

\draw (0,0.25) arc (180:90:0.75);
\draw (0.33,0.67)--(0.22,0.78);

\draw[dashed] (1.2,1)--(1.5,1);
\draw[dashed,<-] (1.2,0)--(1.5,0);
\draw[dashed] (0,-0.2)--(0,-0.5);
\draw[dashed,<-] (1,-0.2)--(1,-0.5);
\end{scope}

\begin{scope}[shift={(4.4,0)}]
\foreach \x in {0,1}{
    \foreach \y in {0,1}{
        \node at (\x,\y)[fill=black,circle,scale=0.5]{};
        }
    }
\node at (0.5,0.5)[fill=black,circle,scale=0.5]{};
\node at (0,1)[fill=white,circle,scale=0.4]{};

\draw[->] (1-0.2,0)--(0.2,0)[];
\draw[->] (1,0.25)--(1,0.75)[];
\draw (0,0.25) arc (180:90:0.75);

\node at (0.61,0.61)[scale=0.8]{$Y$};
\node at (0.89,0.89)[scale=0.8]{$Y$};

\draw[thin](0.4,0.6)--(0.9,1.1)arc(135:-45:0.14)--(0.6,0.4)arc(-45:-225:0.141);

\draw[dashed] (1.2,1)--(1.5,1);
\draw[dashed,<-] (1.2,0)--(1.5,0);
\draw[dashed] (0,-0.2)--(0,-0.5);
\draw[dashed,<-] (1,-0.2)--(1,-0.5);

\begin{scope}[shift={(2,0)}]
\foreach \x in {0,1}{
    \foreach \y in {0,1}{
        \node at (\x,\y)[fill=black,circle,scale=0.5]{};
        }
    }
\node at (0.5,0.5)[fill=black,circle,scale=.5]{};
\node at (0,1)[fill=white,circle,scale=.4]{};

\draw[->] (1-0.2,0)--(0.2,0)[];
\draw[->] (1,0.25)--(1,0.75)[];
\draw (0,0.25) arc (180:90:0.75);

\begin{scope}[shift={(-0.5,-0.5)}]
\node at (0.61,0.61)[scale=0.8]{$Y$};
\node at (0.89,0.89)[scale=0.8]{$X$};

\draw[thin](0.4,0.6)--(0.9,1.1)arc(135:-45:0.14)--(0.6,0.4)arc(-45:-225:0.141);
\end{scope}

\draw[dashed] (1.2,1)--(1.5,1);
\draw[dashed,<-] (1.2,0)--(1.5,0);
\draw[dashed] (0,-0.2)--(0,-0.5);
\draw[dashed,<-] (1,-0.2)--(1,-0.5);
\end{scope}
\node at (1.75,-1){$(b)$};
\end{scope}

\node at (1.75,-1){$(a)$};

\end{tikzpicture}
\end{center}
\caption{Lattice modification to create weight-2 single Majorana/hole operators. (a) The change in edge operators (b) the Majorana/hole operators on the new corner sites. For a corner face where the arrows are all pointing in the other direction, the action of the new edge operator and the new Majorana operators on the vertex qubits will be $X$.}
\label{fig:CornerShaving}
\end{figure}
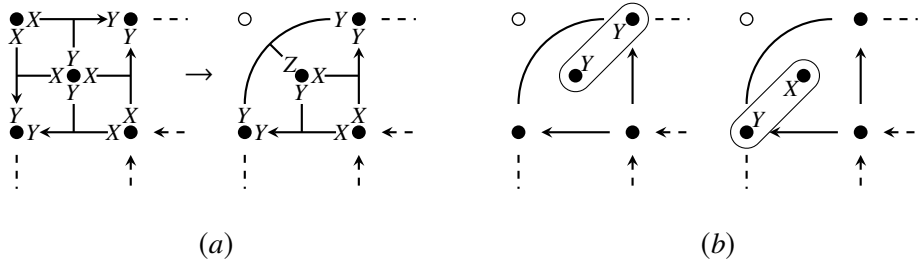

\subsection{Partial Correction of Detectable $\mathbf X$ and $\mathbf Y$ Errors}

In both the encodings detailed in this paper, $X$ and $Y$ errors on primary sites are detectable---but they are not distinguishable as they differ only by $Z$, itself an undetectable logical error.
This means that neither error is correctable since it is impossible to know which correction to apply and the application of the wrong correction leads to a $Z$ error.
Normally this would mean that the codes do not lend themselves to active error correction throughout a circuit run.

However, one can disregard the distinction between $X$ or $Y$ errors and apply a random correction---i.e.\ an $X$ or $Y$.
This yields a $50\%$ chance that the error is corrected; and otherwise the error will be mapped to an undetectable $Z$ error which maps onto natural fermionic phase noise (cf.\ \cref{sec:natural-noise}).
Together with the fact that all single qubit Pauli errors on auxiliary qubits are distinguishable in both the VC and \DK{} encodings, this means that active error correction \emph{can} be used for all single qubit errors, at the expense of introducing additional phase noise on the simulated fermionic system.

\section{Periodic Boundary Conditions}\label{sec:periodic}
As already discussed in \cref{sec:ParityError}, boundary effects require special attention: they introduce weight one Majorana or Majorana hole operators, and furthermore introduce error configurations for which some of the stabilizer violations from \cref{fig:weight-3 syndromes} can cancel, if two or more errors are present.

Much like the surface code emerged from Kitaev's Toric code, we can \emph{also} introduce periodic boundary conditions for the \DK{} encoding.
This yields a uniform periodic lattice if and only if the number of faces (which equals the number of lattice vertices) is even in both directions.
Given such a torus, and by inspection of the error syndromes in \cref{fig:weight-3 syndromes}, the following claim is immediately evident.
\begin{proposition}
On a torus with an even face count $\ge 4$ in both directions, any two weight-1 Pauli errors are either completely detectable; a combination of at most one detectable error and one fermionic phase noise term; or two weight-1 fermionic phase noise terms.
\end{proposition}
\begin{proof}
In the case of periodic boundary conditions, the only single-error syndromes are those illustrated in \cref{fig:weight-3 syndromes}, with none being truncated at the boundaries.
The only pairing of error syndromes in \cref{fig:weight-3 syndromes} that cancel are an $X$ and $Y$ error on the same vertex qubit, which simply corresponds to a $Z$ error on a vertex qubit, i.e.\ fermionic phase noise. 
\end{proof}

For a sufficiently large lattice with periodic boundary conditions, all weight-2 errors either correspond to natural simulated noise, or are detectable. Clearly then, in the absence of periodic boundaries, undetectable weight-2 errors that do not correspond to phase noise must only appear on the boundaries.

\section{Outlook and Discussion}
We see the features and techniques described in this work as being particularly relevant for near term quantum algorithms with no active error correction or fault tolerance \cite{EpsilonCircuits}. Since fermionic encodings are indispensable for fermionic simulation, they constitute a significant fixed overhead in representing any fermionic systems on NISQ devices. One may not be able to afford any additional overhead for supplementary error detection. In this context the time scale of any coherent quantum evolution would have to be upper bounded to ensure that the probability of an error is $\ll 1$. Individual runs might then be post-selected based on whether errors are detected.

Take for example a quantum simulation done by Trotterizing $\HFH$ from \cref{sec:lattice-models}~\cite{Lloyd1996}. This amounts to breaking up the time evolution
\begin{equation}\label{eq:trotter-circuit}
    \exp(\ii t \HFH) = \prod_{i} \exp(\ii t_i h_i) + \BigO(\epsilon)
\end{equation}
for some target error $\epsilon>0$.
Here the $h_i$ are the local Hamiltonian terms, i.e.\ the Pauli product operators that define the chosen encoding; and the $t_i=t_i(t)$ are times chosen such that the given product formula evolves the system under $\HFH$ for time $t$.

Suppose for the sake of argument that the error model is such that errors only occur between Trotter steps, and not in the circuit decomposition of these steps. Given that syndrome measurements are done by measuring stabilizers, which commute with all terms in the Hamiltonian, any sufficiently spatially distant weight-1 Pauli errors in the volume of the computation can be detected, and those computations may be post-selected away. However if two errors occur within the volume of the computation which cancel their respective syndromes, then these runs can not be post-selected, and will contribute to the overall expected accuracy of the computation. Naturally, the question arises how we might address these higher order errors within this framework.

A natural extension of the work presented here is to explore how higher Pauli-weight errors map under these encodings. In particular, if higher-weight errors also map to natural local fermionic noise, then it may be possible to similarly mitigate qubit noise to an even higher order. For example undetectable weight-2 operators in the \DK{} encoding correspond to boundary terms, the majority of which appear as paired hopping terms. This suggests a physical error model in which the simulated system is coupled to a superconducting system at the boundary.

One interesting feature of these results is that there is an inbuilt preference for a particular choice of weight-1 Pauli errors. Thus the work presented here may be especially applicable in cases were the hardware is already biased towards certain Pauli errors.

It is worth noting that if stabilizers are only measured at the end of a run, then, depending on the observables one is interested in measuring, the cost of some stablizer measurements may be significantly reduced, since they can be performed destructively and non-coherently. For example, if one is purely interested in measuring fermion density in the system, then the highly non-local parity operator can be measured simply by measuring every qubit in the $Z$ basis. Thus allowing one to detect the Majorana errors described in this work, without having to resort to any modifications of the encodings.

This work remains relevant even if one does have fault tolerance. Mitigating a large fraction of errors already one level above the error-correcting code would allow a reduction in overhead.

\section*{Acknowledgements}
The authors would like to thank Laura Clinton for helpful discussions.

\printbibliography

\appendix
\newcommand\kk{\mathbf k}
\newcommand{\q}{\mathbf q}
\newcommand{\x}{\mathbf x}
\newcommand{\N}{\Vec N}
\newcommand{\B}{\mathcal{B}}

\section{Derivation of Fermionic Phase Noise}\label{appendix:phaseNoise}
To illustrate the fact that phase noise can be considered natural fermionic noise, we present a derivation of fermionic phase noise for a natural fermionic system. Here we consider a fermionic lattice model coupled to a bosonic system, for example phonons on the lattice. The interaction is mediated through the transfer of momentum; a coupling
known as Fröhlich Hamiltonian \cite{Frohlich1954,Pavarini2017}, which reads
\begin{equation}
   H_\text{int}=\sum_{\kk,\sigma}\sum_\q g_\q a^\dagger_{\kk +\q ,\sigma}a_{\kk ,\sigma}\left(b_\q +b^\dagger_{-\q }\right) .
\end{equation}
Assume for simplicity a momentum-independent coupling $g_q=g:=1$.\footnote{A more sophisticated and realistic analysis could be performed for a momentum-dependent coupling, but we will not attempt this here. 
} A Fourier transform of the fermionic and bosonic operators yields
\begin{equation}
    H_\text{int} =  \sum_{\x} N_{\x} (b_\x + b_\x^\dagger),
\end{equation}
where $N_{\x}=\sum_{\sigma} N_{\x,\sigma}$  and   $N_{\x,\sigma}$ is the fermionic number operator for spin $\sigma\in\{\uparrow,\downarrow\}$ on site $\x$.
Given an interaction strength $\gamma$, the interaction unitary is $U_\text{int} = \ee^{-\ii \gamma H_{\text{int}}} .$

The effective channel on the spin-up fermionic system is then given by
\begin{equation}
    \Lambda(\rho_{\mathcal{F}\uparrow})=\text{Tr}_{\mathcal{B}, \mathcal{F}\downarrow}\left[U_\text{int}(\rho_{\mathcal{F}\uparrow}\otimes \bar{I}_{\mathcal{F}\downarrow}\otimes  \rho_\mathcal{B})U_\text{int}^\dagger\right].
\end{equation}
$\bar{I}_{\mathcal{F}\downarrow}$ is a maximally mixed state on the spin-down fermionic system,
and $\rho_\mathcal{B}$ is a finite temperature Gibbs state with inverse temperature $\beta$ on the bosonic system, which we assumed to be in a completely thermalized bath configuration.
Expanding in the Fock basis $\ket*{\vec{N}_B}$ of $\mathcal{B}$ and $\ket*{\vec{N}_F}$ of $\mathcal{F}_\downarrow$ we retrieve
\begin{equation*}
    \Lambda(\rho_{\mathcal{F}\uparrow})=\frac{1}{d_{\mathcal{F}\downarrow} }\sum_{\vec{N}_F,\vec{N}'_F} \sum_{\vec{N}_B,\vec{N}'_B} \frac{e^{-\beta \vert N_B'\vert }}{Z} \bra*{\vec{N}_B, \vec{N}_F}U_\text{int}\ket*{\vec{N}_B', \vec{N}_F'}\rho_{\mathcal{F}\uparrow}\bra*{\vec{N}_B', \vec{N}_F'}U_\text{int}^\dagger\ket*{\vec{N}_B, \vec{N}_F},
\end{equation*}
where $d_{\mathcal F\downarrow}$ denotes the dimension of the fermionic spin down subsystem.
We note that for a bosonic Fock state $\bra*{\vec{N}_B} H_{\text{int}} \ket*{\vec{N}_B} =0$, such that a second order expansion in $\gamma$ yields
\begin{align*}
    &\Lambda(\rho_{\mathcal{F}\uparrow})= \; \BigO(\gamma^3) + \rho_{\mathcal{F}\uparrow} \nonumber\\
    &+ \frac{\gamma^2}{d_{\mathcal{F}\downarrow} }\sum_{\vec{N}_F,\vec{N}'_F}\sum_{\vec{N}_B,\vec{N}'_B} \frac{e^{-\beta \vert N_B'\vert }}{Z}\bra*{\vec{N}_B, \vec{N}_F}H_\text{int}\ket*{\vec{N}_B', \vec{N}_F'}\rho_{\mathcal{F}\uparrow}\bra*{\vec{N}_B', \vec{N}_F'}H_\text{int}\ket*{\vec{N}_B, \vec{N}_F} \nonumber \\
    &- \frac{\gamma^2}{2d_{\mathcal{F}\downarrow} } \sum_{\vec{N}_F}\sum_{\vec{N}_B}\frac{e^{-\beta \vert N_B\vert }}{Z}\left(\bra*{\vec{N}_B, \vec{N}_F}H_\text{int}^2\ket*{\vec{N}_B, \vec{N}_F}\rho_{\mathcal{F}\uparrow} +\rho_{\mathcal{F}\uparrow} \bra*{\vec{N}_B, \vec{N}_F}H_\text{int}^2\ket*{\vec{N}_B, \vec{N}_F}\right).
\end{align*}
Expanding $H_\text{int}$ and isolating bosonic terms:
\begin{align}
    \Lambda(\rho_{\mathcal{F}\uparrow})=& \; \BigO(\gamma^3) + \rho_{\mathcal{F}\uparrow} \nonumber\\
    &+ \frac{\gamma^2}{d_{\mathcal{F}\downarrow} }\sum_{\vec{N}_F,\vec{N}'_F}\sum_{\x \x'} \Gamma_{\x\x'}(\beta) \bra*{ \vec{N}_F}N_\x\ket*{\vec{N}_F'}\rho_{\mathcal{F}\uparrow}\bra*{ \vec{N}_F'}N_{\x'}\ket*{\vec{N}_F} \nonumber \\
    &- \frac{\gamma^2}{2d_{\mathcal{F}\downarrow} } \sum_{\vec{N}_F}\sum_{\x\x'}\Omega_{\x\x'}(\beta)\bra*{ \vec{N}_F}N_\x N_{\x'}\ket*{ \vec{N}_F}\rho_{\mathcal{F}\uparrow} +\rho_{\mathcal{F}\uparrow} \bra*{ \vec{N}_F}N_\x N_{\x'}\ket*{\vec{N}_F}
\end{align}
where
$$\Gamma_{\x\x'}(\beta)=\sum_{\vec{N}_B,\vec{N}'_B} \bra*{\vec{N}_B}(b_\x+ b_\x^\dagger)\ket*{\vec{N}'_B}\bra*{\vec{N}_B'}(b_{\x'}+ b_{\x'}^\dagger)\ket*{\vec{N}_B} \frac{e^{-\beta \vert N_B'\vert }}{Z}  $$
and
\begin{align} \Omega_{\x\x'}(\beta) &= \sum_{\vec{N}_B} \bra*{\vec{N}_B} (b_\x + b_\x^\dagger)(b_{\x'}+b_{\x'}^\dagger) \ket*{\vec{N}_B}\frac{e^{-\beta \vert N_B\vert }}{Z} \\
&=\sum_{\vec{N}_B,\vec{N}'_B}\bra*{\vec{N}_B'}(b_{\x'}+ b_{\x'}^\dagger)\ket*{\vec{N}_B}\bra*{\vec{N}_B}(b_\x+ b_\x^\dagger)\ket*{\vec{N}'_B} \frac{e^{-\beta \vert N_B\vert }}{Z} \\
&= \Gamma_{\x\x'}(\beta).
\end{align}
Noting that $\Gamma_{\x\x'}(\beta) = \delta_{\x\x'} \Gamma_{\x\x}(\beta) = \Gamma(\beta)\delta_{\x\x'}$, we get
\begin{align}
    \Lambda(\rho_{\mathcal{F}\uparrow})= \;& \BigO(\gamma^3) + \rho_{\mathcal{F}\uparrow} \nonumber\\
    &+ \frac{\gamma^2}{d_{\mathcal{F}\downarrow} }\Gamma(\beta)\sum_{\vec{N}_F,\vec{N}'_F}\sum_{\x}  \bra*{ \vec{N}_F}N_\x\ket*{\vec{N}_F'}\rho_{\mathcal{F}\uparrow}\bra*{ \vec{N}_F'}N_{\x}\ket*{\vec{N}_F} \nonumber \\
    &- \frac{\gamma^2}{2d_{\mathcal{F}\downarrow} } \Gamma(\beta)\sum_{\vec{N}_F}\sum_{\x}\bra*{ \vec{N}_F}N_\x^2\ket*{ \vec{N}_F}\rho_{\mathcal{F}\uparrow} +\rho_{\mathcal{F}\uparrow} \bra*{ \vec{N}_F}N_\x^2\ket*{\vec{N}_F}.
\end{align}
Given that we only work in a fermionic Fock basis $\ket*{\vec N_F}$ on the spin-down sector, we have
\begin{equation}
    \bra*{ \vec{N}_F}N_\x\ket*{\vec{N}_F'} = (N_{\x,\uparrow}+ [\vec{N}_F]_\x) \delta_{N_F, N_F'}
\end{equation}
where $[\vec{N}_F]_\x$ is the $\x^\text{th}$ element of the vector $\vec{N}_F$.
We therefore retrieve
\begin{align}
    \Lambda(\rho_{\mathcal{F}\uparrow})= \; &\BigO(\gamma^3) + \rho_{\mathcal{F}\uparrow} \nonumber\\
    &+ \frac{\gamma^2}{d_{\mathcal{F}\downarrow} }\Gamma(\beta)\sum_{\vec{N}_F}\sum_{\x} (N_{\x,\uparrow}+ [\vec{N}_F]_\x)\rho_{\mathcal{F}\uparrow}(N_{\x,\uparrow}+ [\vec{N}_F]_\x) \nonumber \\
    &- \frac{\gamma^2}{2d_{\mathcal{F}\downarrow} } \Gamma(\beta)\sum_{\vec{N}_F}\sum_{\x}(N_{\x,\uparrow}+ [\vec{N}_F]_\x)^2\rho_{\mathcal{F}\uparrow} +\rho_{\mathcal{F}\uparrow}(N_{\x,\uparrow}+ [\vec{N}_F]_\x)^2.
\end{align}
Noting that $\sum_{\vec{N}_F} [\vec{N}_F]_\x = d_{\mathcal{F}\downarrow}/2$ (as we sum over all possible configurations; and a mode is occupied half of the time) and $\sum_{\vec{N}_F} (N_{\x \uparrow}+ [\vec{N}_F]_\x)^2 = d_F (2N_{\x \uparrow} +1/2) $
and expanding terms yields
\begin{align}
    \Lambda(\rho_{\mathcal{F}\uparrow})=& \; \BigO(\gamma^3) + \rho_{\mathcal{F}\uparrow} \nonumber\\
    &+ \gamma^2\Gamma(\beta) \sum_{\x}  N_{\x,\uparrow}\rho_{\mathcal{F}\uparrow}N_{\x,\uparrow}\nonumber + \frac{1}{2}(N_{\x,\uparrow}\rho_{\mathcal{F}\uparrow} + \rho_{\mathcal{F}\uparrow}  N_{\x,\uparrow} ) +\frac{1}{2} \rho_{\mathcal{F}\uparrow} \nonumber \\
    &- \frac{\gamma^2\Gamma(\beta)}{2} \sum_{\x}(2N_{\x \uparrow} +1/2) \rho_{\mathcal{F}\uparrow} +\rho_{\mathcal{F}\uparrow}(2N_{\x \uparrow} +1/2)
    \nonumber\\
    =& \; \BigO(\gamma^3) + \rho_{\mathcal{F}\uparrow}
    +\gamma^2\Gamma(\beta) \sum_{\x}\left(  N_{\x,\uparrow}\rho_{\mathcal{F}\uparrow}N_{\x,\uparrow}- \frac{1}{2}(N_{\x,\uparrow}\rho_{\mathcal{F}\uparrow} + \rho_{\mathcal{F}\uparrow}  N_{\x,\uparrow})\right).
\end{align}
Now consider the phase operator, defined as $\phi_{\x \uparrow} = (1 - 2 N_{\x \uparrow})$.
Then
\begin{equation}
     \sum_{\x} \left( N_{\x,\uparrow}\rho_{\mathcal{F}\uparrow}N_{\x,\uparrow}- \frac{1}{2}(N_{\x,\uparrow}\rho_{\mathcal{F}\uparrow} + \rho_{\mathcal{F}\uparrow}  N_{\x,\uparrow})\right) = \frac{1}{4}  \sum_\x \left(\phi_{\x \uparrow} \rho_{\mathcal{F}\uparrow}  \phi_{\x \uparrow} - \rho_{\mathcal{F}\uparrow}\right)
\end{equation}
and thus
\begin{align}
    \Lambda(\rho_{\mathcal{F}\uparrow})=& \; \left(1- \frac{\gamma^2 \Gamma(\beta)M}{4}\right)\rho_{\mathcal{F}\uparrow}
    +\frac{\gamma^2 \Gamma(\beta)M}{4} \sum_{\x}  \frac{1}{M}\phi_{\x,\uparrow}\rho_{\mathcal{F}\uparrow}\phi_{\x,\uparrow},
\end{align}
where $M$ is the number of modes.

\end{document}